\documentclass[12pt]{article}

\usepackage{amsmath}
\usepackage{amsthm}
\usepackage{amssymb}

\theoremstyle{plain}
\newtheorem{theorem}{Theorem}
\newtheorem{lemma}[theorem]{Lemma}
\newtheorem{proposition}[theorem]{Proposition}
\newtheorem{corollary}[theorem]{Corollary}

\theoremstyle{definition}
\newtheorem{definition}[theorem]{Definition}

\theoremstyle{remark}
\newtheorem{remark}[theorem]{Remark}

\begin{document}

\title{\textbf{Smoothness of the Gap Function \\
in the BCS-Bogoliubov Theory of Superconductivity}}

\author{Shuji Watanabe\\
Division of Mathematical Sciences\\
Graduate School of Engineering, Gunma University\\
4-2 Aramaki-machi, Maebashi 371-8510, Japan\\
e-mail: watanabe@fs.aramaki.gunma-u.ac.jp}


\date{}

\maketitle

\begin{abstract}
We deal with the gap equation in the BCS-Bogoliubov theory of superconductivity, where the gap function is a function of the temperature $T$ only. We show that the squared gap function is of class $C^2$ on the closed interval $[\,0,\,T_c\,]$. Here, $T_c$ stands for the transition temperature. Furthermore, we show that the gap function is monotonically decreasing on $[0,\,T_c]$ and obtain the behavior of the gap function at $T=T_c$. We mathematically point out some more properties of the gap function.
\end{abstract}

\bigskip

\noindent \textbf{I. INTRODUCTION}

\medskip

Since the surprising discovery by Onnes that the electrical resistivity of mercury drops to zero below the temperature 4.2 K in 1911, the zero electrical resistivity is observed in many metals and alloys. Such a phenomenon is called superconductivity. In 1957 Bardeen, Cooper and Schrieffer \cite{bcs} proposed the highly successful quantum theory of superconductivity, called the BCS theory. In 1958 Bogoliubov \cite{bogoliubov} obtained the results similar to those in the BCS theory using the canonical transformation called the Bogoliubov transformation. The theory by Bardeen, Cooper, Schrieffer and Bogoliubov is called the BCS-Bogoliubov theory.

As an experimental fact, it is observed that it takes a finite energy to excite a quasi particle from the superconducting ground state to an upper energy state. This energy gap is described in terms of the gap function and results from the existence of the electron pairs called the Cooper pairs. Let $k_B>0$ and $\omega_D>0$ stand for the Boltzmann constant and the Debye frequency, respectively. We denote Planck's constant by $h \; (>0)$ and set $\hslash=h/(2\pi)$. Let the temperature $T$ satisfy $0\leq T\leq T_c$, where $T_c>0$ is called the transition temperature (the critical temperature). Let $m>0$ and $\mu>0$ stand for the electron mass and the chemical potential, respectively. Let $k\in\mathbb{R}^3$ denote wave vector and set $\xi_k=\hslash^2|k|^2/(2m)-\mu$. The gap function, denoted by $\Delta_k(T) \,(\geq 0)$, is a function both of the temperature $T$ and of wave vector $k\in\mathbb{R}^3$. In the BCS-Bogoliubov theory, the gap function satisfies the following nonlinear equation called the gap equation:
\begin{equation}\label{eq:descrete}
\Delta_k(T)=-\frac{1}{\,2\,}\sum_{k'}
\frac{U_{k,\,k'}\,\Delta_{k'}(T)}{\,\sqrt{\,\xi_{k'}^2+\Delta_{k'}(T)^2\,}\,}\tanh \frac{\, \sqrt{\,\xi_{k'}^2+\Delta_{k'}(T)^2\,}\,}{2k_BT}
\end{equation}
for $0\leq T\leq T_c$. Here, $k'\in\mathbb{R}^3$ denotes wave vector and the potential $U_{k,\,k'}$ is a function of $k$ and $k'$ satisfying $U_{k,\,k'}\leq 0$. In this connection, see \cite{watanabe} for a new gap equation of superconductivity.

The sum in \eqref{eq:descrete} is often replaced by an integral, and accordingly the gap equation is often regarded as a nonlinear integral equation. In such a situation, Odeh \cite{odeh} and Billard and Fano \cite{billardfano} established the existence and uniqueness of the positive solution to the gap equation in the case $T=0$. In the case $T\geq 0$, Vansevenant \cite{vansevesant} and Yang \cite{yang} determined the transition temperature and showed that there is a unique positive solution to the gap equation. Recently Hainzl, Hamza, Seiringer and Solovej \cite{hhss}, and Hainzl and Seiringer \cite{haizlseiringer} proved that the existence of a positive solution to the gap equation is equivalent to the existence of a negative eigenvalue of a certain linear operator to show the existence of a transition temperature.

Suppose that $U_{k,\,k'}$ is given by (see \cite{bcs})
\begin{equation}\label{eq:uzero}
U_{k,\,k'}=\left\{ \begin{array}{ll}
\displaystyle{-U_0}& \quad (|\xi_k| \leq \hslash\omega_D \;\;
 \mbox{and} \;\; |\xi_{k'}| \leq \hslash\omega_D),\cr
\noalign{\vskip0.2cm}
\displaystyle{0}& \quad (\mbox{otherwise}),
\end{array} \right.
\end{equation}
where $U_0>0$ is a constant. Then $\Delta_k(T)$ depends only on the temperatur $T$ when $|\xi_k|\leq \hslash\omega_D$, whereas $\Delta_k(T)=0$ when $|\xi_k|>\hslash\omega_D$. Let $|\xi_k|\leq \hslash\omega_D$. Then \eqref{eq:descrete} leads to
\begin{equation}\label{eq:descreteprime}
1=\frac{U_0}{\,2\,}\sum_{k'\; (|\xi_{k'}|\leq \hslash\omega_D)}\,
\frac{1}{\,\sqrt{\,\xi_{k'}^2+\Delta(T)^2\,}\,}\tanh \frac{\, \sqrt{\,\xi_{k'}^2+\Delta(T)^2\,}\,}{2k_BT}.
\end{equation}
Here the symbol $k'\; (|\xi_{k'}|\leq \hslash\omega_D)$ stands for $k'$ satisfying $|\xi_{k'}|\leq \hslash\omega_D$, and the gap function $\Delta_k(T)$ is denoted by $\Delta(T)$ simply because it does not depend on $k$ when $k$ satisfies $|\xi_k|\leq \hslash\omega_D$. Accordinly, in this case, the gap function $\Delta(T)$ becomes a function of the temperature $T$ only.

We now replace the sum in \eqref{eq:descreteprime} by the following integral (see \cite{bcs}):
\begin{equation}\label{eq:gapequation}
1=\frac{U_0N_0}{\,2\,}\int_{-\hslash\omega_D}^{\hslash\omega_D} \frac{1}{\,\sqrt{\,\xi^2+\Delta(T)^2\,}\,}\tanh \frac{\, \sqrt{\,\xi^2+\Delta(T)^2\,}\,}{2k_BT}\,d\xi,
\end{equation}
where $0\leq T\leq T_c$, and $N_0>0$ stands for the density of states per unit energy at the Fermi surface.

The gap equation of the form \eqref{eq:gapequation} as well as the hypothesis \eqref{eq:uzero} is accepted widely in condensed matter physics (see e.g. \cite{bcs} and \cite[(11.45), p.392]{ziman}). In this paper we deal with the gap equation \eqref{eq:gapequation} to discuss smoothness of the squared gap function $\Delta(T)^2$ as well as its properties. We show that the squared gap function is of class $C^2$ on the closed interval $[\,0,\,T_c\,]$. Furthermore, we show that the gap function is monotonically decreasing on $[0,\,T_c]$ and obtain the behavior of the gap function at $T=T_c$. We mathematically point out some more properties of the gap function.

It is well known that superconductivity occurs at temperatures below the temperature $T_c>0$ called the transition temperature. Let us now define it.
\begin{definition}[\cite{bcs}]
The transition temperature is the temperature $T_c>0$ satisfying
\[
\frac{1}{\, U_0N_0\,}=\int_0^{\displaystyle{\hslash\omega_D/(2k_BT_c)}} \frac{\,\tanh \eta\,}{\eta}\,d\eta\,.
\]
\end{definition}

\begin{remark}
The equality in the definition above is rewritten as
\[
1=\frac{U_0N_0}{\,2\,}\int_{-\hslash\omega_D}^{\hslash\omega_D}
\frac{1}{\,\sqrt{\,\xi^2\,}\,}\tanh \frac{\, \sqrt{\,\xi^2\,}\,}{2k_BT_c}
\,d\xi\,,
\]
which is obtained by setting $\Delta(T)=0$ and $T=T_c$ in \eqref{eq:gapequation}.
\end{remark}

The paper proceeds as follows. In section 2 we state our main results without proof. In sections 3 and 4 we study some properties of the function $F$ defined by (\ref{eq:functionF}) below. In section 5, on the basis of this study, we prove our main results in a sequence of lemmas.

\bigskip

\noindent \textbf{II. MAIN RESULTS}

\medskip

Let
\[
h(T,\,Y,\,\xi)=\left\{ \begin{array}{ll}\displaystyle{
 \frac{1}{\,\sqrt{\,\xi^2+Y\,}\,}
 \tanh \frac{\, \sqrt{\,\xi^2+Y\,}\,}{2k_BT}
 } \quad &(0<T\leq T_c\,,\quad Y\geq 0),\\
\noalign{\vskip0.3cm} \displaystyle{
 \frac{1}{\,\sqrt{\,\xi^2+Y\,}\,}
 } \quad &(T=0, \quad Y>0)
\end{array}\right.
\]
and set
\begin{equation}\label{eq:functionF}
F(T,\,Y)=\int_0^{\hslash\omega_D} h(T,\,Y,\,\xi)\,d\xi
 -\frac{1}{\,U_0N_0\,}\,.
\end{equation}

We consider the function $F$ on the following domain $W\subset \mathbb{R}^2$:
\[
W=W_1\cup W_2\cup W_3\cup W_4\,,
\]
where
\begin{eqnarray}\nonumber
W_1&=&\left\{ (T,\,Y)\in\mathbb{R}^2:\; 0<T<T_c\,,\; 0<Y<2\,\Delta_0^2 \right\},\\ \nonumber
W_2&=&\left\{ (0,\,Y)\in\mathbb{R}^2:\; 0<Y<2\,\Delta_0^2 \right\},\\
 \nonumber
W_3&=&\left\{ (T,\,0)\in\mathbb{R}^2:\; 0<T\leq T_c \right\},\\ \nonumber
W_4&=&\left\{ (T_c\,,\,Y)\in\mathbb{R}^2:\; 0<Y<2\,\Delta_0^2 \right\}.
\end{eqnarray}
Here,
\begin{equation}\label{eq:delta0}
\Delta_0=\frac{\hslash\omega_D}{\,\sinh\frac{1}{\,U_0N_0\,}\,}\,.
\end{equation}

\begin{remark} The gap equation (\ref{eq:gapequation}) is rewritten as $\displaystyle{F(T,\,Y)=0}$, where $Y$ corresponds to $\displaystyle{ \Delta(T)^2}$.
\end{remark}

The following are our main results.
\begin{theorem}\label{thm:solution}
Let $F$ be as in (\ref{eq:functionF}) and $\Delta_0$ as in (\ref{eq:delta0}). Then there is a unique solution: $T \mapsto Y=f(T)$ of class $C^2$ on the closed interval $[\,0,\,T_c\,]$ to the gap equation $\displaystyle{F(T,\,Y)=0}$ such that the function $f$ is monotonically decreasing on $[0,\,T_c]$:
\[
f(0)=\Delta_0^2>f(T_1)>f(T_2)>f(T_c)=0, \qquad 0<T_1<T_2<T_c\,.
\]
\end{theorem}

Let $g$ be given by
\begin{equation}\label{eq:fng}
g(\eta)= \left\{ \begin{array}{ll}\displaystyle{
\frac{1}{\,\eta^2\,}\left( \frac{1}{\,\cosh^2\eta \,}
 -\frac{\,\tanh\eta\,}{\eta}\right) } \qquad &(\eta>0),\\
\noalign{\vskip0.3cm} \displaystyle{
-\frac{\,2\,}{\,3\,} } &(\eta=0).
\end{array}\right.
\end{equation}
Note that $g(\eta)<0$, as is pointed out by Lemma \ref{gproperty} below. Let $G$ be
given by
\begin{equation}\label{eq:fncg}
G(\eta)= \left\{ \begin{array}{ll}\displaystyle{
\frac{1}{\,\eta^2\,} \left\{ 3\,g(\eta)
+2\,\frac{\tanh\eta}{\,\eta\,\cosh^2\eta \,} \right\} } \qquad &(\eta>0),\\
\noalign{\vskip0.3cm} \displaystyle{
-\frac{\,16\,}{\,15\,} } &(\eta=0).
\end{array}\right.
\end{equation}
See Lemma \ref{lm:cgproperty} below for some properties of $G$.

\begin{proposition}\label{prp:behavior}
Let $f$ be as in Theorem \ref{thm:solution}. Then the values of the derivative
$f'$ at $T=0$ and at $T=T_c$ are given as follows:
\[
f'(0)=0,\qquad f'(T_c)=8\,k_B^2T_c\,\frac{\,\displaystyle{
\int_0^{\displaystyle{\hslash\omega_D/(2k_BT_c)}}
 \frac{d\eta}{\,\cosh^2\eta\,}
}\,}{\,\displaystyle{
\int_0^{\displaystyle{\hslash\omega_D/(2k_BT_c)}} g(\eta)\,d\eta
}\,}<0\,.
\]
Consequently, the behavior of $f$ at $T=T_c$ is given by
\begin{equation}\label{eq:behavioroff}
f(T) \approx -f'(T_c)\,(T_c-T)
= - \, 8\,k_B^2\,\frac{\,\displaystyle{
\int_0^{\displaystyle{\hslash\omega_D/(2k_BT_c)}} \frac{d\eta}{\,\cosh^2\eta\,}
}\,}{\,\displaystyle{ 
\int_0^{\displaystyle{\hslash\omega_D/(2k_BT_c)}} g(\eta)\,d\eta
}\,}\; T_c\,(T_c-T)\,.
\end{equation}
\end{proposition}

\begin{remark} The behavior of $f$ similar to (\ref{eq:behavioroff}) was already obtained by a different method in the context of theoretical, condensed matter physics. However, Proposition \ref{prp:behavior} gives the new form (\ref{eq:behavioroff}) explicitly in the context of mathematics.
\end{remark}

Let $\phi$ be a function of $\eta$ and let $\eta$ be a function of $\xi$.
Set
\begin{equation}\label{eq:integralI}
I\left[ \,\phi(\eta) \,\right]=\int_0^{\hslash\omega_D} \phi(\eta)\,d\xi\,.
\end{equation}
\begin{proposition}\label{prp:behaviorprime}
Let $f$ be as in Theorem \ref{thm:solution} and $I\left[ \cdot \right]$ as in (\ref{eq:integralI}). Then the values of the second derivative $f''$ at $T=0$
and at $T=T_c$ are given as follows:
\[
f''(0)=0,
\]
\begin{eqnarray}\nonumber
& &f''(T_c) \\ \nonumber
&=& 16\, k_B^2\,\frac{\,\displaystyle{
 I\left[ \,\frac{\,\eta_0\tanh \eta_0-1\,}{\cosh^2 \eta_0} \,\right]
}\,}{\, \displaystyle{
 I\left[ \,g(\eta_0)\, \right]  }\,}
-32\, k_B^2\,\frac{\,\displaystyle{
 I\left[ \,\frac{1}{\,\cosh^2 \eta_0\,}\, \right]
 I\left[ \,\frac{\,\tanh \eta_0\,}{\,\eta_0\cosh^2 \eta_0\,}\, \right]
 }\,}{\, \displaystyle{
 \left\{ \, I\left[ \,g(\eta_0)\, \right] \, \right\}^2  }\,} \\ \nonumber
& &\quad +8\, k_B^2\,\frac{\,\displaystyle{
 \left\{ I\left[ \,\frac{1}{\,\cosh^2 \eta_0\,}\, \right] \right\}^2
 I\left[ \,G(\eta_0)\, \right]
 }\,}{\, \displaystyle{
 \left\{ \, I\left[ \,g(\eta_0)\, \right] \, \right\}^3 }\,}\,,
\qquad \eta_0=\frac{\xi}{\,2k_BT_c\,}\,.
\end{eqnarray}
\end{proposition}

Combining Theorem \ref{thm:solution} with Propositions \ref{prp:behavior} and
\ref{prp:behaviorprime} immediately implies the following.
\begin{corollary}
There is a unique gap function: $T \mapsto \Delta(T)=\sqrt{f(T)}$ on the closed
interval $[\,0,\,T_c\,]$ such that it is of class $C^2$ on the interval
$[\,0,\,T_c\,)$, and is monotonically decreasing on $[0,\,T_c]$:
\[
\Delta(0)=\Delta_0>\Delta(T_1)>\Delta(T_2)>\Delta(T_c)=0, \qquad
0<T_1<T_2<T_c\,.
\]
Furthermore, \quad $\displaystyle{ \Delta'(0)=\Delta''(0)=0}$ \quad and \quad
$\displaystyle{ \lim_{T\uparrow T_c} \Delta'(T)=-\infty }$.
\end{corollary}

\bigskip

\noindent \textbf{III. THE FIRST-ORDER PARTIAL DERIVATIVES OF $F$}

\medskip

In this section we deal with the first-order partial derivatives of the
function $F$ and show that $F$ is of class $C^1$ on $W$.

A straightforward calculation yields the following.
\begin{lemma}\label{gproperty}
Let $g$ be as in (\ref{eq:fng}). Then the function $g$ is of class $C^1$
on $[0,\,\infty)$ and satisfies
\[
g(\eta)<0,\qquad g'(0)=0,\qquad
\lim_{\eta\to\infty}g(\eta)=\lim_{\eta\to\infty}g'(\eta)=0.
\]
\end{lemma}

\begin{lemma}\label{lm:FTFYW1}
The values of the partial derivatives $\displaystyle{\frac{\,\partial F\,}{\,\partial T\,}}$ and $\displaystyle{\frac{\,\partial F\,}{\,\partial Y\,}}$ exist at each point in $W_1$.
\end{lemma}

\begin{proof}
Let $(T,\,Y)\in W_1$. Then there is a $\theta$ \  \  $(0<\theta<1)$ satisfying $\theta T_c<T<T_c$. Therefore,
\[
\left| \frac{\,\partial h\,}{\,\partial T\,}(T,\,Y,\,\xi) \right| \leq \frac{1}{\,2k_B\theta^2T_c^2\,}\,,
\]
where the right side is integrable on $[0,\,\hslash\omega_D]$. Hence the value
of $\displaystyle{\frac{\,\partial F\,}{\,\partial T\,}}$ exists at each point in $W_1$. Here,
\begin{equation}\label{eq:FT}
\frac{\,\partial F\,}{\,\partial T\,}(T,\,Y)=-\frac{1}{\,2k_BT^2\,}
 \int_0^{\hslash\omega_D} \frac{d\xi}{\,\cosh^2\eta\,}\,,
 \qquad \eta=\frac{\,\sqrt{\,\xi^2+Y\,}\,}{2k_BT}\,.
\end{equation}
On the other hand,
\[
\frac{\,\partial h\,}{\,\partial Y\,}(T,\,Y,\,\xi)=\frac{g(\eta)}{\,2\,(2k_BT)^3\,}\,,\qquad \eta=\frac{\,\sqrt{\,\xi^2+Y\,}\,}{2k_BT}\,.
\]
Hence, by Lemma \ref{gproperty},
\[
\left| \frac{\,\partial h\,}{\,\partial Y\,}(T,\,Y,\,\xi) \right|\leq
\frac{\,\displaystyle{ \max_{\eta\geq 0}|g(\eta)| }\,}{\,2\,(2k_B\theta T_c)^3\,}\,,
\]
where the right side is also integrable on $[0,\,\hslash\omega_D]$. Hence the value of $\displaystyle{\frac{\,\partial F\,}{\,\partial Y\,}}$ exists at each point in $W_1$. Here,
\begin{equation}\label{eq:FY}
\frac{\,\partial F\,}{\,\partial Y\,}(T,\,Y)=\frac{1}{\,2(2k_BT)^3\,}
 \int_0^{\hslash\omega_D} g(\eta)\, d\xi\,,
 \qquad \eta=\frac{\,\sqrt{\,\xi^2+Y\,}\,}{2k_BT}\,.
\end{equation}
\end{proof}

\begin{lemma}\label{lm:FTFYW}
The values of the partial derivatives $\displaystyle{\frac{\,\partial F\,}{\,\partial T\,}}$ and $\displaystyle{\frac{\,\partial F\,}{\,\partial Y\,}}$ exist at each point in $W$.
\end{lemma}

\begin{proof}
We show that the values of $(\partial F/\partial T)$ and $(\partial F/\partial Y)$ exist at each point in $W_2$. Let $(0,\,Y_0)\in W_2$. Then
\begin{eqnarray}\nonumber
\left| \frac{\,F(T,\,Y_0)-F(0,\,Y_0)\,}{T} \right|
&\leq& \int_0^{\hslash\omega_D} \frac{1}{\,T\sqrt{\,\xi^2+Y_0\,}\,}
 \left( 1-\tanh\frac{\,\sqrt{\,\xi^2+Y_0\,}\,}{2k_BT} \right)\,d\xi \\ \nonumber&\leq& \frac{\,4k_B^2T\,}{Y_0} \int_0^{\hslash\omega_D}
 \frac{d\xi}{\,\sqrt{\,\xi^2+Y_0\,}\,}\,,
\end{eqnarray}
and hence
\[
\frac{\,\partial F\,}{\,\partial T\,}(0,\,Y_0)=0.
\]
On the other hand, for $Y>Y_0/2$,
\[
\frac{\,F(0,\,Y)-F(0,\,Y_0)\,}{Y-Y_0}=-\int_0^{\hslash\omega_D}
 \frac{d\xi}{\,\sqrt{\xi^2+Y}\sqrt{\xi^2+Y_0}
 \left( \sqrt{\xi^2+Y}+\sqrt{\xi^2+Y_0} \right)\,}\,.
\]
Note that
\[
\frac{1}{\,\sqrt{\xi^2+Y}\sqrt{\xi^2+Y_0}
 \left( \sqrt{\xi^2+Y}+\sqrt{\xi^2+Y_0} \right)\,}\leq
 \frac{\,2(\sqrt{2}-1)\,}{Y_0^{3/2}},
\]
where the right side is integrable on $[0,\,\hslash\omega_D]$. Therefore,
\begin{equation}\label{eq:FYonW}
\frac{\,\partial F\,}{\,\partial Y\,}(0,\,Y_0)
 =-\frac{1}{\,2\,}\int_0^{\hslash\omega_D}
 \frac{d\xi}{\,( \sqrt{\xi^2+Y_0} )^3\,}
 =-\frac{\hslash\omega_D}{\,2\,Y_0\sqrt{\hslash^2\omega_D^2+Y_0}\,}\,.
\end{equation}

Similarly we can show that those exist at each point in $W_3$, and in $W_4$. Their values are given as follows:\quad For $(T_0,\,0)\in W_3$,
\begin{eqnarray}\nonumber
\frac{\,\partial F\,}{\,\partial T\,}(T_0,\,0)
&=&-\frac{1}{\,2k_BT_0^2\,}\int_0^{\hslash\omega_D}
 \frac{d\xi}{\,\cosh^2 \frac{\xi}{\,2k_BT_0\,}\,},\\ \nonumber
\frac{\,\partial F\,}{\,\partial Y\,}(T_0,\,0)
&=&\frac{1}{\,2(2k_BT_0)^3\,}\int_0^{\hslash\omega_D}
 g\left( \frac{\xi}{\,2k_BT_0\,} \right)\, d\xi\,,
\end{eqnarray}
and for $(T_c\,,\,Y_0)\in W_4$,
\begin{eqnarray}\nonumber
\frac{\,\partial F\,}{\,\partial T\,}(T_c\,,\,Y_0)
&=&-\frac{1}{\,2k_BT_c^2\,}\int_0^{\hslash\omega_D}
 \frac{d\xi}{\,\cosh^2 \frac{\,\sqrt{\xi^2+Y_0}\,}{\,2k_BT_c\,}\,},\\ \nonumber
\frac{\,\partial F\,}{\,\partial Y\,}(T_c\,,\,Y_0)
&=&\frac{1}{\,2(2k_BT_c)^3\,}\int_0^{\hslash\omega_D}
 g\left( \frac{\,\sqrt{\xi^2+Y_0}\,}{\,2k_BT_c\,} \right)\, d\xi\,.
\end{eqnarray}
The result follows.
\end{proof}

Lemmas \ref{gproperty}, \ref{lm:FTFYW1} and \ref{lm:FTFYW} immediately give the following.
\begin{lemma}\label{lm:FTFYminus}\quad At each $(T,\,Y)\in W\setminus W_2$,
\[
\frac{\,\partial F\,}{\,\partial T\,}(T,\,Y)<0,\qquad
\frac{\,\partial F\,}{\,\partial Y\,}(T,\,Y)<0.
\]
\end{lemma}

We now study the continuity of the functions $F$, $(\partial F/\partial T)$ and $(\partial F/\partial Y)$ on $W$.

\begin{lemma}\label{lm:FC1onW1}
The partial derivatives $\displaystyle{\frac{\,\partial F\,}{\,\partial T\,}}$ and $\displaystyle{\frac{\,\partial F\,}{\,\partial Y\,}}$ are continuous on $W_1$. Consequently, the function $F$ is of class $C^1$ on $W_1$.
\end{lemma}

\begin{proof}
It is enough to show that the functions: $(T,\,Y)\mapsto I_1(T,\,Y)$ and $(T,\,Y)\mapsto I_2(T,\,Y)$ (see (\ref{eq:FT}) and (\ref{eq:FY})) are continuous at $(T_0,\,Y_0)\in W_1$. Here,
\begin{equation}\label{eq:I1I2}
I_1(T,\,Y)=\int_0^{\hslash\omega_D} \frac{d\xi}{\,\cosh^2\eta\,}\,,
\quad I_2(T,\,Y)=\int_0^{\hslash\omega_D} g(\eta)\, d\xi\,,\quad
\eta=\frac{\,\sqrt{\,\xi^2+Y\,}\,}{2k_BT}\,.
\end{equation}
Set $\displaystyle{\eta_0=\frac{\,\sqrt{\,\xi^2+Y_0\,}\,}{2k_BT_0}}$.\quad
Since $(T,\,Y)\in W_1$ is close to $(T_0,\,Y_0)\in W_1$, it follows that
$T>T_0/2$. Then
\begin{eqnarray}\nonumber
& &\left| I_1(T,\,Y)-I_1(T_0,\,Y_0)\right| \\ \nonumber
&\leq& \int_0^{\hslash\omega_D}
 \left|\left( \frac{1}{\,\cosh\eta\,}+\frac{1}{\,\cosh\eta_0\,} \right)
 \frac{\,\cosh\eta-\cosh\eta_0\,}{\,\cosh\eta\,\cosh\eta_0\,} \right|
 \,d\xi \\ \nonumber
&\leq& 2\hslash\omega_D\sinh
 \frac{\,\sqrt{\,\hslash^2\omega_D^2+2\,\Delta_0^2\,}\,}{k_BT_0}
 \left( \frac{\,\sqrt{\,\hslash^2\omega_D^2+2\,\Delta_0^2\,}\,}{k_BT_0^2}|T-T_0|+\frac{|Y-Y_0|}{\,k_BT_0\sqrt{Y_0}\,}\right), \\ \nonumber
& &\left| I_2(T,\,Y)-I_2(T_0,\,Y_0)\right| \\ \nonumber
&\leq& \int_0^{\hslash\omega_D} \left| g(\eta)-g(\eta_0) \right|\,d\xi
 \\ \nonumber
&\leq& \hslash\omega_D\,\displaystyle{ \max_{\eta\geq 0}|g'(\eta)| }
 \left( \frac{\,\sqrt{\,\hslash^2\omega_D^2+2\,\Delta_0^2\,}\,}{k_BT_0^2}|T-T_0|+\frac{|Y-Y_0|}{\,k_BT_0\sqrt{Y_0}\,}\right).
\end{eqnarray}
Thus the functions: $(T,\,Y)\mapsto I_1(T,\,Y)$ and $(T,\,Y)\mapsto I_2(T,\,Y)$, and hence $(\partial F/\partial T)$ and $(\partial F/\partial Y)$ are continuous at $(T_0,\,Y_0)\in W_1$.
\end{proof}

\begin{lemma}
\quad The function $F$ is continuous on $W$.
\end{lemma}

\begin{proof}
Note that $F$ is continuous on $W_1$ by Lemma \ref{lm:FC1onW1}. We then show that $F$ is continuous on $W_2$.

Let $(0,\,Y_0)\in W_2$ and let $(T,\,Y)\in W_1\cup W_2$. Since $(T,\,Y)$ is close to $(0,\,Y_0)$, it follows that $Y>Y_0/2$. Then, by (\ref{eq:functionF}),
\begin{eqnarray}\nonumber
& &\left| F(T,\,Y)-F(0,\,Y_0)\right| \\ \nonumber
&\leq& \int_0^{\hslash\omega_D} \left\{
 \frac{\,1-\tanh\frac{\,\sqrt{\,\xi^2+Y\,}\,}{2k_BT}\,}{\sqrt{\,\xi^2+Y_0\,}}+\left| \frac{1}{\,\sqrt{\,\xi^2+Y\,}\,}
 -\frac{1}{\,\sqrt{\,\xi^2+Y_0\,}\, } \right| \right\}\,d\xi
 \\ \nonumber
&\leq& \hslash\omega_D \left\{
 \frac{1}{\,\sqrt{Y_0}\,}
 \left( 1-\tanh\frac{\,\sqrt{\,Y_0/2\,}\,}{2k_BT} \right)
 +\frac{\,2\,|Y-Y_0|\,}{\,(\sqrt{2}+1)Y_0^{3/2}\,} \right\}.
\end{eqnarray}
Thus $F$ is continuous on $W_2$. Similarly we can show the continuity of $F$ on $W_3$, and on $W_4$.
\end{proof}

\begin{lemma}\label{lm:FC1onW}
The partial derivatives $\displaystyle{\frac{\,\partial F\,}{\,\partial T\,}}$ and $\displaystyle{\frac{\,\partial F\,}{\,\partial Y\,}}$ are continuous on $W$. Consequently, the function $F$ is of class $C^1$ on $W$.
\end{lemma}

\begin{proof}
Note that $(\partial F/\partial T)$ and $(\partial F/\partial Y)$ are continuous on $W_1$ by Lemma \ref{lm:FC1onW1}. We then show that $(\partial F/\partial T)$ and $(\partial F/\partial Y)$ are continuous at $(T_c\,,\,0)\in W_3$. We can show their continuity at other points in $W$ similarly.

\textit{Step 1}. Let $(T,\,Y)\in W_1$. We show
\[
\frac{\,\partial F\,}{\,\partial T\,}(T,\,Y) \to
 \frac{\,\partial F\,}{\,\partial T\,}(T_c\,,\,0),\;\;
\frac{\,\partial F\,}{\,\partial Y\,}(T,\,Y) \to
\frac{\,\partial F\,}{\,\partial Y\,}(T_c\,,\,0) \quad \mbox{as}\;
(T,\,Y) \to (T_c\,,\,0).
\]
Since $(T,\,Y)$ is close to $(T_c\,,\,0)$, it then follows that $T_c/2<T<T_c$. Set $\eta_0=\frac{\,\sqrt{\hslash^2\omega_D^2+2\,\Delta_0^2}\,}{k_BT_c}$. Then
\begin{eqnarray}\nonumber
& & \left| \frac{1}{\,T^2\cosh^2 \frac{\,\sqrt{\xi^2+Y}\,}{\,2k_BT\,}\,}
-\frac{1}{\,T_c^2\cosh^2 \frac{\,\xi\,}{\,2k_BT_c\,} \,} \right| \\ \nonumber
&\leq& \frac{\,8\cosh \eta_0\,}{T_c^3} \left\{ \left| T-T_c \right|
 \left( \cosh \eta_0+\eta_0 \sinh \eta_0 \right)
 +\frac{\,\sqrt{Y}\,}{4k_B} \sinh \eta_0 \right\},
\end{eqnarray}
and hence \quad $\displaystyle{
(\partial F/\partial T)(T,\,Y)-(\partial F/\partial T)(T_c\,,\,0) \to 0 \quad \mbox{as}\quad (T,\,Y) \to (T_c\,,\,0)}$.

Since
\[
\left| g\left( \frac{\,\sqrt{\,\xi^2+Y\,}\,}{2k_BT} \right)
 -g\left( \frac{\,\xi\,}{\,2k_BT_c\,} \right) \right|
\leq \max_{\eta\geq 0}|g'(\eta)| \left(
 \frac{\,\hslash\omega_D\left| T-T_c \right|\,}{k_BT_c^2}+
 \frac{\,\sqrt{Y}\,}{\,k_BT_c\,} \right),
\]
it follows that
\[
\int_0^{\hslash\omega_D} \left\{
 g\left( \frac{\,\sqrt{\,\xi^2+Y\,}\,}{2k_BT} \right)
 -g\left( \frac{\,\xi\,}{\,2k_BT_c\,} \right) \right\}\,d\xi \to 0 \quad
 \mbox{as}\quad (T,\,Y) \to (T_c\,,\,0),
\]
and hence \quad $\displaystyle{
(\partial F/\partial Y)(T,\,Y)-(\partial F/\partial Y)(T_c\,,\,0) \to 0 \quad \mbox{as}\quad (T,\,Y) \to (T_c\,,\,0)}$.

\textit{Step 2}. When $(T,\,Y)=(T,\,0)\in W_3$ and $(T,\,Y)=(T_c\,,\,Y)\in W_4$, an argument similar to that in Step 1 gives
\[
\frac{\,\partial F\,}{\,\partial T\,}(T,\,0) \to
 \frac{\,\partial F\,}{\,\partial T\,}(T_c\,,\,0),\quad
\frac{\,\partial F\,}{\,\partial Y\,}(T,\,0) \to
\frac{\,\partial F\,}{\,\partial Y\,}(T_c\,,\,0) \quad \mbox{as}\quad
(T,\,0) \to (T_c\,,\,0)
\]
and
\[
\frac{\,\partial F\,}{\,\partial T\,}(T_c\,,\,Y) \to
 \frac{\,\partial F\,}{\,\partial T\,}(T_c\,,\,0),\quad
\frac{\,\partial F\,}{\,\partial Y\,}(T_c\,,\,Y) \to
\frac{\,\partial F\,}{\,\partial Y\,}(T_c\,,\,0)
\]
as \quad $(T_c\,,\,Y) \to (T_c\,,\,0)$. The result follows.
\end{proof}

\bigskip

\noindent \textbf{IV. THE SECOND-ORDER PARTIAL DERIVATIVES OF $F$}

\medskip

In this section we deal with the second-order partial derivatives of the function $F$ and show that $F$ is of class $C^2$ on $W_1$.

A straightforward calculation yields the following.
\begin{lemma}\label{lm:cgproperty}
Let $G$ be as in (\ref{eq:fncg}) and $g$ as in (\ref{eq:fng}). Then the
function $G$ is of class $C^1$ on $[0,\,\infty)$ and satisfies
\[
g'(\eta)=-\eta G(\eta),\qquad G'(0)=0,\qquad
\lim_{\eta\to\infty}G(\eta)=\lim_{\eta\to\infty}G'(\eta)=0.
\]
\end{lemma}

\begin{lemma}\label{lm:FTTexistence} The values of the partial derivatives \  
$\displaystyle{ \frac{\,\partial^2 F\,}{\,\partial T^2\,} }$,\quad
$\displaystyle{ \frac{\partial}{\,\partial Y\,}\left( 
\frac{\,\partial F\,}{\,\partial T\,} \right) }$, \\
$\displaystyle{ \frac{\partial}{\,\partial T\,}\left( 
\frac{\,\partial F\,}{\,\partial Y\,} \right) }$ \  and \   
$\displaystyle{ \frac{\,\partial^2 F\,}{\,\partial Y^2\,} }$
exist at each point in $W_1$. Furthermore,
\[
\frac{\partial}{\,\partial Y\,}\left( \frac{\,\partial F\,}{\,\partial T\,}\right)=\frac{\partial}{\,\partial T\,}\left( 
\frac{\,\partial F\,}{\,\partial Y\,} \right) \qquad \mbox{on} \quad W_1\,.
\]
\end{lemma}

\begin{proof}
Let $(T,\,Y)\in W_1$. Then there is a $\theta$ \  \  $(0<\theta<1)$ satisfying $\theta T_c<T<T_c$. Set $\displaystyle{ \eta=\frac{\,\sqrt{\,\xi^2+Y\,}\,}{2k_BT} }$. Then
\[
\left| \frac{\partial}{\,\partial T\,}\left( \frac{1}{\,\cosh^2\eta\,} \right)
\right| \leq \frac{\,\sqrt{\hslash^2\omega_D^2+2\,\Delta_0^2}\,}{k_B\theta^2T_c^2}\,,
\]
where the right side is integrable on $[0,\,\hslash\omega_D]$. So, the function: $(T,\,Y)\mapsto I_1(T,\,Y)$ (see (\ref{eq:I1I2})), and hence $(\partial F/\partial T)$ (see (\ref{eq:FT})) is differentiable with respect to $T$ on $W_1$, and the second-order partial derivative is given by
\[
\frac{\,\partial^2 F\,}{\,\partial T^2\,}(T,\,Y)
=\frac{1}{\,k_BT^3\,} \left\{ I_1(T,\,Y)-\int_0^{\hslash\omega_D}
 \frac{\eta\,\tanh\eta}{\,\cosh^2\eta \,}\,d\xi \right\},\qquad
 \eta=\frac{\,\sqrt{\,\xi^2+Y\,}\,}{2k_BT}\,.
\]

Similarly we can show that $(\partial F/\partial T)$ is differentiable with respect to $Y$ on $W_1$, that $(\partial F/\partial Y)$ is differentiable with respect to $T$ on $W_1$, and that $(\partial F/\partial Y)$ is differentiable with respect to $Y$ on $W_1$. The corresponding second-order partial derivatives are given as follows:
\[
\frac{\partial}{\,\partial Y\,}\left( 
\frac{\,\partial F\,}{\,\partial T\,} \right)(T,\,Y)
=\frac{\partial}{\,\partial T\,}\left( 
\frac{\,\partial F\,}{\,\partial Y\,} \right)(T,\,Y)
=\frac{1}{\,(2k_BT)^3T\,} \int_0^{\hslash\omega_D}
 \frac{\tanh\eta}{\,\eta\,\cosh^2\eta \,}\,d\xi,
\]
\[
\frac{\,\partial^2 F\,}{\,\partial Y^2\,}(T,\,Y)
=-\,\frac{1}{\,4\,(2k_BT)^5\,} \int_0^{\hslash\omega_D} G(\eta)\,d\xi\,,
\qquad \eta=\frac{\,\sqrt{\,\xi^2+Y\,}\,}{2k_BT}\,.
\]
Here, $G$ is that in Lemma \ref{lm:cgproperty} (see also (\ref{eq:fncg})).
\end{proof}

\begin{lemma}\label{lm:FC2W1}
The partial derivatives
$\displaystyle{ \frac{\,\partial^2 F\,}{\,\partial T^2\,} }$,
$\displaystyle{ \frac{\partial}{\,\partial Y\,}\left(
 \frac{\,\partial F\,}{\,\partial T\,} \right) }$ and
$\displaystyle{\frac{\,\partial^2 F\,}{\,\partial Y^2\,}}$
are continuous on $W_1$. Consequently, $F$ is of class $C^2$ on $W_1$.
\end{lemma}

\begin{proof}
We show that $(\partial^2 F/\partial Y^2)$ is continuous on $W_1$. Similarly we can show the continuity of other second-order partial derivatives.

By the form of $(\partial^2 F/\partial Y^2)$ given in the proof of Lemma \ref{lm:FTTexistence}, it suffices to show that the function: $(T,\,Y)\mapsto I_3(T,\,Y)$ is continuous at $(T_0\,,\,Y_0)\in W_1$. Here,
\[
I_3(T,\,Y)=\int_0^{\hslash\omega_D} G(\eta)\,d\xi\,,\qquad
\eta=\frac{\,\sqrt{\,\xi^2+Y\,}\,}{2k_BT}\,.
\]
Since $(T,\,Y)$ is close to $(T_0\,,\,Y_0)$, it then follows that $T_0/2<T$.
A straightforward calculation then gives
\begin{eqnarray}\nonumber
& &\left| \,I_3(T,\,Y)-I_3(T_0\,,\,Y_0)\,\right|\\ \nonumber
&\leq& \hslash\omega_D\,\displaystyle{\max_{\eta\geq 0}|G'(\eta)|}
\left( \frac{\,\sqrt{\,\hslash^2\omega_D^2+2\,\Delta_0^2\,}\,}{k_BT_0^2}|T-T_0|+\frac{|Y-Y_0|}{\,k_BT_0\sqrt{Y_0}\,}\right).
\end{eqnarray}
Hence the function: $(T,\,Y)\mapsto I_3(T,\,Y)$ is continuous at $(T_0\,,\,Y_0)\in W_1$.
\end{proof}

\bigskip

\noindent \textbf{V. PROOFS OF OUR MAIN RESULTS}

\medskip

In this section we prove Theorem \ref{thm:solution}, Propositions
\ref{prp:behavior} and \ref{prp:behaviorprime} in a sequence of lemmas.

\begin{remark} One may prove the theorem and the propositions above on the basis of the implicit function theorem. In this case, \textit{an interior point} $(T_0,\,Y_0)$ of the domain $W$ satisfying $\displaystyle{ F(T_0,\,Y_0)=0 }$ need to exist. But there are the two points $(0,\,\Delta_0^2)$ and $(T_c\,,\,0)$ in \textit{the boundary} of $W$ satisfying
\begin{equation}\label{eq:F=0}
F(0,\,\Delta_0^2)=F(T_c\,,\,0)=0.
\end{equation}
So one can not apply the implicit function theorem in its present form.
\end{remark}

\begin{lemma}\label{lm:existence}
There is a unique solution: $T \mapsto Y=f(T)$ to the gap equation $\displaystyle{F(T,\,Y)=0}$ such that the function $f$ is continuous on the closed interval
$[0,\,T_c]$ and satisfies $f(0)=\Delta_0^2$ and $f(T_c)=0$.
\end{lemma}

\begin{proof}
By Lemmas \ref{lm:FTFYminus}, \ref{lm:FC1onW} and (\ref{eq:F=0}), the function: $Y \mapsto F(T_c\,,\,Y)$ is monotonically decreasing and there is a $Y_1$ \quad $(0<Y_1<2\Delta_0^2)$ satisfying $F(T_c\,,\,Y_1)<0$. Note that $Y_1$ is arbitrary as long as $0<Y_1<2\Delta_0^2$. Hence, by Lemma \ref{lm:FC1onW}, there is a $T_1$ \quad $(0<T_1<T_c)$ satisfying $F(T_1\,,\,Y_1)<0$. Hence, $F(T,\,Y_1)<0$ for $T_1\leq T\leq T_c$. On the other hand, by Lemmas \ref{lm:FTFYminus}, \ref{lm:FC1onW} and (\ref{eq:F=0}), the function: $T \mapsto F(T,\,0)$ is monotonically decreasing and there is a $T_2$ \quad $(0<T_2<T_c)$ satisfying $F(T_2\,,\,0)>0$. Note that $T_2$ is arbitrary as long as $0<T_2<T_c$. Hence, $F(T,\,0)>0$ for $T_2\leq T<T_c$.

Let $\max(T_1\,,\,T_2)\leq T<T_c$ and fix $T$. It then follows from Lemmas \ref{lm:FTFYminus} and \ref{lm:FC1onW} that the function: $Y \mapsto F(T,\,Y)$ with $T$ fixed is monotonically decreasing on $[0,\,Y_1]$. Since $F(T,\,0)>0$ and $F(T,\,Y_1)<0$, there is a unique $Y$\   $(0<Y<Y_1)$ satisfying $F(T,\,Y)=0$. When $T=T_c$, there is a unique value $Y=0$ satisfying $F(T_c\,,\,Y)=0$ \  (see (\ref{eq:F=0})).

Since $F$ is continuous on $W$ by Lemma \ref{lm:FC1onW}, there is a unique solution: $T \mapsto Y=f(T)$ to the gap equation $\displaystyle{F(T,\,Y)=0}$ such that $f$ is continuous on $[\max(T_1\,,\,T_2),\,T_c]$ and $f(T_c)=0$.

Since $(\partial F/\partial Y)(0,\,Y)<0$ \  $(0<Y<2\Delta_0^2)$ by (\ref{eq:FYonW}), there is a unique value $Y=\Delta_0^2$ satisfying $F(0,\,Y)=0$.
Combining Lemma \ref{lm:FC1onW} with Lemma \ref{lm:FTFYminus} therefore implies that the function $f$ is continuous on $[0,\,T_c]$ and that $f(0)=\Delta_0^2$ and $f(T_c)=0$.
\end{proof}

\begin{lemma}\label{lm:derivative}
The function $f$ given by Lemma \ref{lm:existence} is of class $C^1$ on $[0,\,T_c]$, and the derivative $f'$ satisfies
\[
f'(0)=0,\qquad f'(T_c)=8\, k_B^2T_c\,\frac{\,\displaystyle{
 \int_0^{\displaystyle{\hslash\omega_D/(2k_BT_c)}} \frac{d\xi}{\,\cosh^2\xi\,} }\,}
{\,\displaystyle{ 
\int_0^{\displaystyle{\hslash\omega_D/(2k_BT_c)}} g(\xi)\,d\xi }\,}\,.
\]
\end{lemma}

\begin{proof}
Lemmas \ref{lm:FTFYW1}, \ref{lm:FTFYW}, \ref{lm:FTFYminus} and \ref{lm:FC1onW} immediately imply that the function $f$ is of class $C^1$ on the interval $[0,\,T_c]$ and that its derivative is given by
\begin{equation}\label{eq:f'(T)}
f'(T)=-\frac{\,F_T(T,\,f(T))\,}{\,F_Y(T,\,f(T))\,}\,.
\end{equation}
The values of $f'(0)$ and $f'(T_c)$ are derived from (\ref{eq:f'(T)}).
\end{proof}

Combining (\ref{eq:f'(T)}) with Lemma \ref{lm:FTFYminus} immediately yields the following.
\begin{lemma}
The function $f$ given by Lemma \ref{lm:existence} is monotonically decreasing on $[0,\,T_c]$:
\[
f(0)=\Delta_0^2>f(T_1)>f(T_2)>f(T_c)=0, \qquad 0<T_1<T_2<T_c\,.
\]
\end{lemma}

\begin{lemma}
Let $f$ be as in Lemma \ref{lm:existence} and $I\left[ \cdot \right]$ as in (\ref{eq:integralI}). Then the function $f$ is of class $C^2$ on $[0,\,T_c]$, and the second derivative $f''$ satisfies \quad $\displaystyle{ f''(0)=0 }$
\quad and
\begin{eqnarray}\nonumber
& & f''(T_c) \\ \nonumber
&=& 16\, k_B^2\,\frac{\,\displaystyle{
 I\left[ \,\frac{\,\eta_0\tanh \eta_0-1\,}{\cosh^2 \eta_0}\, \right]
}\,}{\, \displaystyle{
 I\left[ \,g(\eta_0)\, \right]  }\,}
-32\, k_B^2\,\frac{\,\displaystyle{
 I\left[ \,\frac{1}{\,\cosh^2 \eta_0\,}\, \right]
 I\left[ \,\frac{\,\tanh \eta_0\,}{\,\eta_0\cosh^2 \eta_0\,}\, \right]
 }\,}{\, \displaystyle{
 \left\{ \, I\left[ \,g(\eta_0)\, \right] \, \right\}^2  }\,} \\ \nonumber
& &\quad +8\, k_B^2\,\frac{\,\displaystyle{
 \left\{ I\left[ \,\frac{1}{\,\cosh^2 \eta_0\,}\, \right] \right\}^2
 I\left[ \,G(\eta_0)\, \right]
 }\,}{\, \displaystyle{
 \left\{ \, I\left[ \,g(\eta_0)\, \right] \, \right\}^3 }\,}\,,
\qquad \eta_0=\frac{\xi}{\,2k_BT_c\,}\,.
\end{eqnarray}
\end{lemma}

\begin{proof}
Lemma \ref{lm:FC2W1} implies that $f$ is of class $C^2$ on the open interval $(0,\,T_c)$ and that
\begin{equation}\label{eq:f''}
f''(T)=\frac{\,-F_{TT}F_Y^2+2F_{TY}F_TF_Y-F_{YY}F_T^2\,}{F_Y^3}\,,\qquad
0<T<T_c\,.
\end{equation}
So we have only to deal with $f$ and its derivatives at $T=0$ and at $T=T_c$.

\textit{Step 1}. We show that $f'$ is differentiable at $T=0$ and that
$f''$ is continuous at $T=0$.

Note that $f'(0)=0$ by Lemma \ref{lm:derivative}. Since $T$ is close to $T=0$, the inequality $f(T)>\Delta_0^2\, /2$ holds. It then follows from (\ref{eq:f'(T)}), (\ref{eq:FT}) and (\ref{eq:FY}) that
\[
\left| \frac{\,f'(T)-f'(0)\,}{T} \right| \leq
4\,\frac{ \,\displaystyle{
 {\;\sqrt{\hslash^2\omega_D^2+2\,\Delta_0^2}\;}^3
 \exp\left(-\frac{\,\sqrt{\Delta_0^2/2}\,}{k_BT}\right) }
 \,}{ \displaystyle{ 
 k_BT^3\left( \tanh \eta_1-\frac{\eta_1}{\,\cosh^2\eta_1\,}\right) }
 } \to 0 \qquad (T\downarrow 0).
\]
Here, \quad $\displaystyle{ \eta_1=\frac{\,\sqrt{\,\xi_1^2+f(T)\,}\,}{2k_BT}
\to \infty }$ \quad as $T\downarrow 0$ \quad $(0<\xi_1<\hslash\omega_D)$. Hence $f'$ is differentiable at $T=0$ and $f''(0)=0$.

By (\ref{eq:f''}), a similar argument gives \quad $\displaystyle{ 
\lim_{T\downarrow 0} f''(T)=0 }$. Hence $f''$ is continuous at $T=0$.

\textit{Step 2}. We show that $f'$ is differentiable at $T=T_c$ and that $f''$ is continuous at $T=T_c$.

Note that
\[
f'(T_c)=8\, k_B^2T_c\,\frac{\,\displaystyle{
 I\left[ \,\frac{1}{\,\cosh^2 \eta_0\,}\, \right]
}\,}{\,\displaystyle{ 
 I\left[ \,g(\eta_0)\, \right]
}\,}\,,
\qquad \eta_0=\frac{\xi}{\,2k_BT_c\,}
\]
by Lemma \ref{lm:derivative}. It follows from (\ref{eq:f'(T)}), (\ref{eq:FT}) and (\ref{eq:FY}) that
\[
f'(T)=8\, k_B^2T \,\frac{\,\displaystyle{
 I\left[ \,\frac{1}{\,\cosh^2 \eta\,}\, \right]
}\,}{\,\displaystyle{ 
 I\left[ \,g(\eta)\, \right]
}\,}
\,, \qquad \eta=\frac{\,\sqrt{\,\xi^2+f(T)\,}\,}{2k_BT}\,.
\]
Hence
\begin{eqnarray}\nonumber
& & \frac{\,f'(T_c)-f'(T)\,}{T_c-T} \\ \nonumber
&=&8\, k_B^2\,\frac{\,\displaystyle{
 I\left[ \,\frac{1}{\,\cosh^2 \eta_0\,}\, \right]  }\,}{\,\displaystyle{ 
 I\left[ \,g(\eta_0)\, \right]  }\,}+\frac{8\, k_B^2T}{\,T_c-T\,}
\frac{
 I\left[ \,\frac{1}{\,\cosh^2 \eta_0\,}\, \right] \left\{
 I\left[ \,g(\eta)\, \right]-I\left[ \,g(\eta_0)\, \right] \right\}
}{
I\left[ \,g(\eta_0)\, \right] I\left[ \,g(\eta)\, \right]
} \\ \nonumber
& & \quad +\frac{8\, k_B^2T}{\,T_c-T\,}\,
\frac{\,
 I\left[ \,g(\eta_0)\, \right] \left\{
 I\left[ \,\frac{1}{\,\cosh^2 \eta_0\,}\, \right]
 -I\left[ \,\frac{1}{\,\cosh^2 \eta\,}\, \right] \right\}
\,}{
I\left[ \,g(\eta_0)\, \right] \, I\left[ \,g(\eta)\, \right]
}\,.
\end{eqnarray}
Note that \quad $g(\eta)-g(\eta_0)=(\eta-\eta_0)g'(\eta_1)$ and
$\cosh\eta-\cosh\eta_0=(\eta-\eta_0)\sinh\eta_2$. Here,
\[
\eta_0=\frac{\xi}{\,2k_BT_c\,}<\eta_i<\eta=\frac{\,\sqrt{\,\xi^2+f(T)\,}\,}{2k_BT}\,,\qquad i=1,\,2
\]
and
\[
\eta-\eta_0=\frac{1}{\,2k_BT\,}\left\{
\frac{f(T)}{\, \sqrt{\,\xi^2+f(T)\,}+\xi \,}+\xi\frac{\,T_c-T\,}{T_c}
\right\}.
\]
Since $T$ is close to $T_c$, the inequality $T>T_c\, /2$ holds. Therefore,
by Lemma \ref{lm:cgproperty},
\[
\left| \frac{g'(\eta_1)}{\,\sqrt{\,\xi^2+f(T)\,}+\xi\,} \right| \leq
\frac{1}{\,k_BT_c\,}\,\max_{\eta\geq 0}\left| G(\eta) \right|
\]
and
\[
\left| \frac{\sinh\eta_2}{\,\sqrt{\,\xi^2+f(T)\,}+\xi\,} \right| \leq
\frac{1}{\,k_BT_c\,}\,\max_{0 \leq\eta\leq M}
 \left| \frac{\,\sinh\eta\,}{\eta} \right|,\qquad
M=\frac{\,\sqrt{ \hslash^2\omega_D^2+2\,\Delta_0^2 }\,}{k_BT_c}\,.
\]
So $f'$ is differentiable at $T=T_c$, and it is easy to see that the form of $f''(T_c)$ is exactly the same as that mentioned just above.

Furthermore, it follows from (\ref{eq:f''}) that $f''$ is continuous at
$T=T_c$.
\end{proof}


\end{document}